%% file: main.tex
\documentclass[conference,a4paper,10pt]{IEEEtran}

\usepackage{silence}
\usepackage[left=1.52cm,right=1.5cm,top=1.9cm,columnsep=0.63cm]{geometry}

\usepackage{etoolbox}
\patchcmd{\thebibliography}
  {\setlength{\itemsep}{0pt}}
  {\setlength{\itemsep}{-2pt}\setlength{\parsep}{0pt}\setlength{\topsep}{2pt}}
  {}{}

\BeforeBeginEnvironment{theorem}{\vspace{-5pt}}
\AfterEndEnvironment{theorem}{\vspace{-5pt}}
\BeforeBeginEnvironment{lemma}{\vspace{-5pt}}
\AfterEndEnvironment{lemma}{\vspace{-5pt}}
\BeforeBeginEnvironment{corollary}{\vspace{-5pt}}
\AfterEndEnvironment{corollary}{\vspace{-5pt}}
\BeforeBeginEnvironment{proof}{\vspace{-5pt}}
\AfterEndEnvironment{proof}{\vspace{-5pt}}

\usepackage{enumitem}
\setlist{nosep, topsep=1pt, itemsep=0pt, parsep=0pt}

\usepackage{cite}

\usepackage{amsmath,amssymb,amsthm,fixmath}

\usepackage{color,colortbl}
\usepackage{xcolor}

\usepackage{siunitx}
\DeclareSIUnit{\dBm}{dBm}

\usepackage{graphicx}
\usepackage[labelformat=simple]{subcaption}

\usepackage{epstopdf}
\usepackage{cuted}
\usepackage{multirow,booktabs}
\usepackage{diagbox}
\usepackage{makecell}
\usepackage{hhline}
\usepackage{array} 
\newcolumntype{x}{!{\vrule width 2px}}
\newcolumntype{y}{!{\vrule width 1.5px}}

\usepackage{optidef}

\usepackage[shortcuts,acronym,automake]{glossaries}
\input{glossary}

\usepackage{tcolorbox}
\tcbuselibrary{many}
\newtheorem{theorem}{Theorem}
\newtheorem{lemma}{Lemma}

\newtheorem{remark}{Remark}
\usepackage[norelsize,linesnumbered,ruled]{algorithm2e}
\SetKwRepeat{Do}{do}{while}
\makeatletter
\newcommand{\removelatexerror} {\let\@latex@error\@gobble}
\makeatother

\usepackage[normalem]{ulem}
\newcommand{\revise}[2]{{\color{red}\sout{#1}}{\color{blue}#2}}

\usepackage[textsize=tiny,colorinlistoftodos]{todonotes}
\makeatletter
\define@key{todonotes}{bh}[]{
	\setkeys{todonotes}{author=\textbf{Bin}, color=lime!30}}%
\define@key{todonotes}{wc}[]{
	\setkeys{todonotes}{author=\textbf{Wenwen}, color=green!30}}%
\makeatother

\newif\ifreviewmode
\reviewmodefalse

\ifreviewmode
\else
  \renewcommand{\todo}[1]{} 
  \renewcommand{\revise}[2]{#2} 
\fi

\begin{document}

\title{Confusion--Erasure Bounds of Error-Bounded Decoders under QAM}

\author{
	
	\IEEEauthorblockN{
        Wenwen~Chen\IEEEauthorrefmark{1},
		Bin~Han\IEEEauthorrefmark{1},
and~Hans~D.~Schotten\IEEEauthorrefmark{1}\IEEEauthorrefmark{3}
	}
	
	\IEEEauthorblockA{
		\IEEEauthorrefmark{1}RPTU University Kaiserslautern-Landau, 
		\IEEEauthorrefmark{3}German Research Center for Artificial Intelligence (DFKI)
	}
}


\maketitle

\begin{abstract}
6G is expected to push \ac{urllc} toward stringent residual-error targets for mission-critical services, where undetected errors and erasures carry fundamentally different costs.
\Ac{bler} conflates block confusions (undetected errors)
and block erasures, which have fundamentally different impacts on
system reliability. \revise{This paper extends the confusion and erasure
analysis of error-bounded decoders to square \ac{qam} constellations in the \ac{fbl} regime.}
{To characterize  the confusion rate analytically, we consider an error-bounded decoder with spherical decision regions. 
We focus on square \ac{qam} constellations in the \ac{fbl} regime. The confusion rate is governed by 
the Euclidean distances between modulated codewords, whereas the code structure constrains their Hamming distances. 
}
\revise{To handle \ac{qam}'s heterogeneous symbol energies---which make the per-pair Euclidean distance 
a distribution rather than a single value---we derive analytical lower and upper bounds on the block confusion rate by, 
respectively, collapsing this distribution to its \ac{rms} distance and averaging the pairwise confusion over it. These
bounds are proven to be monotonically decreasing in both the average
symbol energy and the blocklength, with the decrease rate governed by the
constellation order.}
{To bridge the two metrics, 
we introduce a typical distance, which represents the Euclidean distance of a codeword pair at a given Hamming distance on average, 
and relate it to the Hamming distance through a per-symbol squared-distance coefficient. 
On this basis, we derive a lower bound on the block confusion rate for random codes and an upper bound 
for codebooks whose minimum Euclidean distance reaches the typical distance at the non-overlapping threshold. We further characterize how both 
bounds evolve with the average symbol energy and the blocklength.} 
Numerical results\revise{}{, validated by an exact computation for short codes,} 
confirm that as the \ac{snr} or redundancy increases, the confusion rate falls many orders of magnitude below the reliability target, 
leaving detectable erasures as the dominant residual error.
\end{abstract}


\IEEEpeerreviewmaketitle

\glsresetall

\section{Introduction}\label{sec:introduction}

As 6G moves beyond 5G's reliability targets, supporting latency- and reliability-sensitive services such as vehicular control and industrial automation at short blocklengths, the accurate characterization of \ac{phy}-layer decoding failures becomes essential.
\Ac{bler} has long served as the primary performance
metric for digital communication systems, particularly in the context 
of \ac{urllc} in the \ac{fbl} regime~\cite{polyanskiy2010, durisi2016}. However, 
\ac{bler} conflates two fundamentally distinct error types: \emph{block 
erasures}, in which the receiver detects the failure and can request retransmission, and 
\emph{block confusions} (undetected errors), in which the decoder 
\revise{silently commits to an incorrect codeword}{outputs an incorrect codeword without being able to recognize it as incorrect}. This distinction is critical in safety-critical applications such as vehicular communications and avionics, 
where undetected control-command errors can be catastrophic, yet rigorous information-theoretic analysis of block confusion in the \ac{fbl} 
regime remains largely absent.

The theoretical foundations of confusion and erasure were established 
by Forney~\cite{forney1968exponential}, with subsequent refinements by Merhav 
\cite{merhav2008} and Somekh-Baruch and Merhav~\cite{somekh2011exact}. 
These results, however, characterize asymptotic exponential decay rates 
and do not directly apply to the \ac{fbl} setting. Systematic \ac{fbl} analysis 
emerged with Polyanskiy et al.~\cite{polyanskiy2010}, and has since 
been extended to fading and multi-antenna channels~\cite{yang2014, 
collins2018coherent}, as well as cooperative and \ac{noma} systems~\cite{hu2015capacity, 
xiang2020}. Erasure channels have been studied extensively for \ac{bec} 
\cite{di2002finite, dana2006}, $q$-ary erasure channels~\cite{liva2013}, 
and block erasure channels~\cite{i2006coding, didier2006new}. Despite 
this, the distinction between block confusion and block erasure in the 
\ac{fbl} regime has not been rigorously analyzed, and cross-layer protocols 
universally model PHY failures as erasures — an assumption lacking 
rigorous PHY-layer validation.

This gap was addressed in our prior work~\cite{han2025confusions}, which derived the first analytical bounds on block confusion and erasure probabilities for error-bounded decoders in the \ac{fbl} regime, confirming that block confusion rates fall orders of magnitude below the \ac{bler} constraint. However, that analysis was restricted to \ac{psk} constellations, where uniform symbol energies place all codewords on a single hypersphere. This excludes $M$-\ac{qam}---the dominant modulation in 5G NR and Wi-Fi 6/7---whose heterogeneous symbol energies distribute codewords across distinct concentric hyperspheres, fundamentally altering the geometric structure of the analysis.

    In this paper, we extend~\cite{han2025confusions} to square $M$-\ac{qam}
constellations. The principal contributions are as follows: 1) 
Unlike the single-hypersphere \ac{psk} case of \cite{han2025confusions}, 
the Euclidean distance of a Hamming-distance-$d$ codeword pair is not a single value for \ac{qam}. 
\revise{To handle this, we derive two bounds on
    the block confusion rate $P_{\mathrm{con}}$: the lower bound collapses the
    distribution to its \ac{rms} distance $\sqrt{\Delta d}$, whereas the upper bound
    averages the pairwise confusion over the full per-shell distribution,
    giving a tight estimate that avoids the loose worst-case nearest-neighbour
    approximation.}
    {We show that, for random codes, it concentrates around the typical distance $\sqrt{\Delta d}$. 
    And we derive the confusion bounds based on this typical distance.
    }
2) \revise{We prove that both bounds are monotonically decreasing in
    the average symbol energy $\bar{E}_s$ and the blocklength $n$, with
    $P_{\mathrm{con}}^{\mathrm{LB}}$ being additionally convex in
    $\bar{E}_s$.}{We analyze the sensitivity of both bounds to the \ac{snr} and the blocklength, and show that they generally decrease as either grows.}
    3) Numerical results confirm that as the \ac{snr} or redundancy increases, 
    the undetected confusion becomes a vanishing fraction of $\varepsilon$\revise{, leaving the remaining failures as erasures detectable
    by the channel code's inherent error-detection capability, such as the parity-check syndrome of \ac{ldpc}, BCH, and Reed--Solomon codes.}
    { and the remaining failures are erasures declared by the decoder. 
    Moreover, an exact computation for short polar, eBCH, and RS codes, validated by Monte-Carlo simulation, also validate our theoretical analysis.
    }
    

\section{Error-bounded Decoder} \label{sec:system}


Since complete decoders such as \ac{ml} and \ac{map} always output a codeword and cannot detect errors, we consider only incomplete 
decoders in this paper. We therefore adopt an \revise{$\varepsilon$-bounded}{error-bounded} decoder \revise{whose non-overlapping decision regions reject list outputs.
The decision region radius
should satisfy $R\leqslant D_{\min}/2$ so that list-decoding does not exist.}{whose decision regions are balls of radius $R$ around the modulated codewords. 
Analogous to Forney's generalized decoder~\cite{forney1968exponential}, whose decision regions remain disjoint with the erasure option and overlap with the list option, we say that the decoder operates in the \emph{unique-decoding regime} if $R\leqslant D_{\min}/2$ and in the \emph{list-decoding regime} if $R>D_{\min}/2$,
where $D_{\min}$ is the minimum Euclidean distance between modulated codewords.
In the former, the decision regions are disjoint and $\mathbf{y}$ lies in at most one of them; in the latter, $\mathbf{y}$ may lie in several decision regions. This paper considers the unique-decoding regime only.}

\revise{Consider a codebook $\mathcal{X}$ over the $M$-\ac{qam} alphabet $\mathcal{M}$, with block length $n = k + r$ ($k$ payload, $r$ redundancy). For a transmitted codeword $\mathbf{x} \in \mathcal{X}$ over an AWGN channel with $\mathbf{w} \sim \mathcal{CN}(\mathbf{0},2\sigma^2 I_n)$, decoding the received $\mathbf{y} = \mathbf{x} + \mathbf{w}$ yields one of three outcomes: 1) \emph{Correct decoding:} $\mathbf{y}$ falls within the decision region of $\mathbf{x}$. 2) \emph{Erasure:} $\mathbf{y}$ falls outside all decision regions. 3) \emph{Confusion:} $\mathbf{y}$ falls within the decision region of $\mathbf{x}' \neq \mathbf{x}$.}
{Consider a codebook $\mathcal{C}\subset\mathcal{Q}^n$ of size $|\mathcal{C}|=M^k$ over the $M$-ary alphabet $\mathcal{Q}$, 
with block length $n=k+r$, where $k$ denotes payload and $r$ denotes redundancy symbols. Each codeword $\mathbf{c}\in\mathcal{C}$ is mapped 
symbol-wise onto the $M$-\ac{qam} constellation $\mathcal{M}\subset\mathbb{C}$ by a bijective labelling $\varphi:\mathcal{Q}\to\mathcal{M}$, 
which yields the modulated codeword $\mathbf{x}=\varphi(\mathbf{c})\in\mathcal{M}^n$. Thus, the set of modulated codewords is $\mathcal{X}=\varphi(\mathcal{C})$. 
Since the decoder compares the received signal with the modulated codewords through the Euclidean distance $\|\mathbf{y}-\mathbf{x}\|$, 
the analysis below is carried out on $\mathcal{X}$. For a transmitted modulated codeword $\mathbf{x}\in\mathcal{X}$ over an AWGN channel 
with $\mathbf{w}\sim\mathcal{CN}(\mathbf{0},2\sigma^2 I_n)$, the received signal is $\mathbf{y}=\mathbf{x}+\mathbf{w}$. 
Let $B_R(\mathbf{x}')=\{\mathbf{y}:\|\mathbf{y}-\mathbf{x}'\|\leqslant R\}$ denote the decision region of $\mathbf{x}'\in\mathcal{X}$. 
The decoder outputs $\hat{\mathbf{x}}=\mathbf{x}'$ if $\mathbf{y}\in B_R(\mathbf{x}')$ and declares an erasure if $\mathbf{y}$ lies 
in no decision region. The decoding is correct if $\hat{\mathbf{x}}=\mathbf{x}$ and results in a confusion if $\hat{\mathbf{x}}=\mathbf{x}'\neq\mathbf{x}$. The radius $R$ is chosen such that $P(\mathbf{y}\notin B_R(\mathbf{x}))=\varepsilon$ for every transmitted $\mathbf{x}$, i.e., $\varepsilon$ bounds the total probability of erasure and confusion.}

\revise{For a spherical codebook $\mathcal{X} \subset \mathcal{M}^n$ under \ac{qam} modulation, 
codewords exhibit varying energy levels owing to the non-uniform symbol energies inherent to the constellation. Consequently, 
the codewords reside on distinct hyperspheres of radii $\sqrt{E_l}$.}
{Under \ac{qam} modulation, the modulated codewords $\mathbf{x}\in\mathcal{X}$ exhibit varying energy levels owing to the non-uniform symbol energies 
inherent to the constellation. Consequently, they reside on distinct hyperspheres of radii $\sqrt{E_l}$, where $E_l$ denotes the $l$-th distinct value of the codeword energy $\|\mathbf{x}\|^2$, as illustrated in Fig.~\ref{fig:energy_shells} 
for two modulated codewords $\mathbf{x}$ and $\mathbf{x}'$ on different hyperspheres. In the unique-decoding regime, their decision regions $B_R(\mathbf{x})$ 
and $B_R(\mathbf{x}')$ are disjoint.}

\begin{figure}[!htpb]
    \centering
    \includegraphics[width=0.58\linewidth]{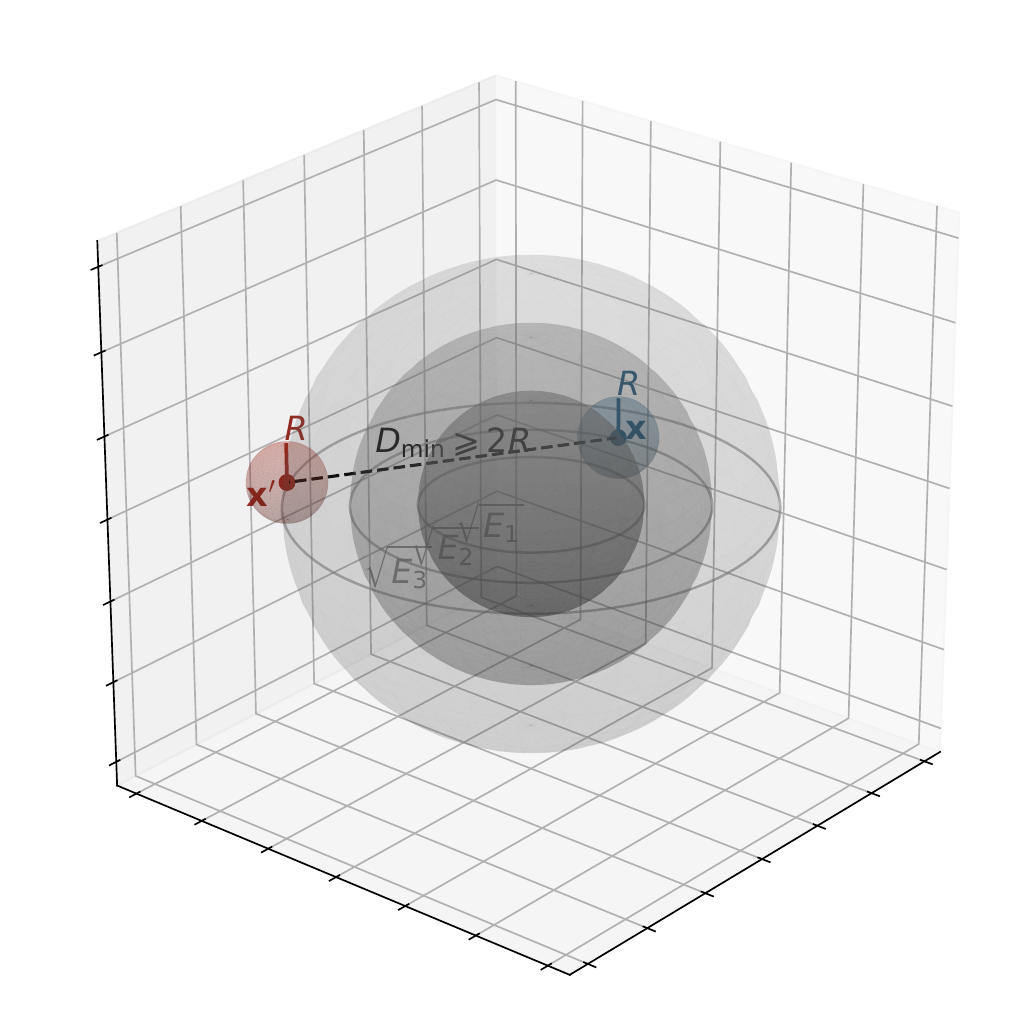}
    \caption{Geometric illustration of the decision balls of radius $R$ centered
at two modulated codewords $\mathbf{x}$ and $\mathbf{x}'$.}
    \label{fig:energy_shells}
\end{figure}


\revise{}{We now derive the block confusion rate $P_{\mathrm{con}}$. To this end, we sort the other codewords $\mathbf{x}'\neq\mathbf{x}$, into which the transmitted $\mathbf{x}$ may be confused, according to the Hamming distance between the underlying codewords, and then evaluate the pairwise confusion rate between $\mathbf{x}$ and each of them.}
For an arbitrary codeword $\mathbf{x} \in \mathcal{X}$ transmitted over the
noisy channel, let $\mathcal{V}(\mathcal{X}, \mathbf{x}, d)$ denote the
set of codewords in $\mathcal{X}$ \revise{at Hamming distance $d$ from $\mathbf{x}$}{whose underlying codewords are at Hamming distance $d$, where $\mathbf{x}=\varphi(\mathbf{c})$ and $\mathbf{x}'=\varphi(\mathbf{c}')$}, i.e., $\mathcal{V}(\mathcal{X}, \mathbf{x}, d) \triangleq \{ \mathbf{x}' \in \mathcal{X} \mid d_H(\mathbf{c}, \mathbf{c}') = d \}$,
and therewith the total set of codewords in $\mathcal{X}$ other than
$\mathbf{x}$, $\mathcal{V}_{\Sigma}(\mathcal{X}, \mathbf{x}) = \bigcup_{d > 0} \mathcal{V}(\mathcal{X}, \mathbf{x}, d)$.

\revise{Consider the decision region centered at $\mathbf{x}$ and with another codeword $\mathbf{x}' \in \mathcal{V}_{\sum}(\mathcal{X},\mathbf{x})$ on arbitrary surface, as shown in Fig. \ref{fig:energy_shells}}{}

\revise{By definition, $\varepsilon = P(\|\mathbf{w}\|\geqslant R)$ encompasses both erasure and confusion. The decision radius is thus:}{Since $\mathbf{y}\notin B_R(\mathbf{x})$ if and only if $\|\mathbf{w}\|>R$, the decision radius is thus} $R(\varepsilon)=F^{-1}_{\|\mathbf{w}\|}(1-\varepsilon)=\sigma F^{-1}_{\chi_{2n}}(1-\varepsilon)$,
where $F_{\chi_{2n}}(x)=P(n;x^2/2)$ is the \ac{cdf} of $\chi$ distribution with $2n$ degrees of freedom and $P(s,x)=\gamma(s,x)/\Gamma(s)$ is the regularized gamma function, with $\gamma(s,x)=\int_0^x t^{s-1}e^{-t}\mathrm{d}t$ and $\Gamma(s)=\gamma(s,\infty)$.
As $R(\varepsilon)$ depends only on $\sigma$, $n$, and $\varepsilon$, not on the codeword, the same radius applies to every codeword.
The pairwise confusion rate is $P(\mathbf{x}\rightarrow\mathbf{x}')\triangleq P(\hat{\mathbf{x}}=\mathbf{x}'|\mathbf{x})=P(\|\mathbf{x}+\mathbf{w}-\mathbf{x}'\|\leqslant R)$,
i.e., the probability that $\mathbf{w}$ falls in $B_R(\mathbf{x}'-\mathbf{x})$. Conditioning on $\|\mathbf{w}\| = w$, the fraction of the sphere $B_w(\mathbf{0})$ inside $B_R(\mathbf{x}'-\mathbf{x})$ is $P[\mathbf{w}\in B_R(\mathbf{x}'-\mathbf{x})\mid\|\mathbf{w}\|=w]=\Omega_{2n}[\theta_w(w,\|\mathbf{x}'-\mathbf{x}\|,R)]$,
where $\theta_w(w,D,R)=\arccos\big((w^2+D^2-R^2)/(2wD)\big)$ is the cap angle and $\Omega_{2n}(\theta)=\int_0^{\theta}(\sin\phi)^{2n-2}\mathrm{d}\phi\big/\!\int_0^{\pi}(\sin\phi)^{2n-2}\mathrm{d}\phi$ is the angle fraction.

Thus $P(\mathbf{x} \to \mathbf{x}') = P_{\mathrm{pair}}(\|\mathbf{x}' - \mathbf{x}\|)$ depends only on the Euclidean distance. Since $\|\mathbf{w}\|/\sigma \sim \chi_{2n}$:
\begin{equation}
\label{eq:P_pair}
    P_{\mathrm{pair}}(D) = \int_{(D-R)/\sigma}^{(D+R)/\sigma} 
    f_{\chi_{2n}}(u)\, \Omega_{2n}\left(\theta_w(\sigma u, D,R)\right)\, \mathrm{d}u.
\end{equation}

Under the assumption that \revise{codewords are uniformly distributed}{all codewords are equally likely to be transmitted, $P(\mathbf{x})=1/|\mathcal{X}|$}, the overall block confusion rate can be expressed as:
\begin{align}
\label{eq:P_con}
\begin{split}
    P_{\mathrm{con}}
    & = \sum_{\mathbf{x} \in \mathcal{X}}
      \left[ P(\mathbf{x})
        \sum_{\mathbf{x}' \in \mathcal{V}_\Sigma(\mathcal{X},\, \mathbf{x})}
        P(\mathbf{x} \to \mathbf{x}')
      \right] \\
   & = \frac{1}{|\mathcal{X}|}
      \sum_{\mathbf{x} \in \mathcal{X}}\;
      \sum_{\mathbf{x}' \in \mathcal{V}_\Sigma(\mathcal{X},\,\mathbf{x})}
      P(\mathbf{x} \to \mathbf{x}').
      \end{split}
\end{align}
Thus, the block erasure rate $P_{\mathrm{ers}}=\varepsilon-P_{\mathrm{con}}$, which implies that the upper bound on the confusion rate is the lower bound on the erasure rate, and vice versa.

\revise{In this work, we want to derive the minimum achievable erasure probability $P_{\mathrm{ers}}$ given fixed $(\varepsilon, \bar{E}_s, \sigma^2)$.}{In this work, we derive bounds on the block confusion rate $P_{\mathrm{con}}$ for a fixed error bound $\varepsilon$, 
average symbol energy $\bar{E}_s$, and noise variance $\sigma^2$. Through $P_{\mathrm{ers}}=\varepsilon-P_{\mathrm{con}}$, 
they directly yield the corresponding bounds on the block erasure rate.} 
In the following analysis, we only consider the case $|\mathcal{X}| = M^k$, corresponding to full utilization of the payload code space.

\section{Bounds Analysis} \label{sec:bounds}
It follows from~\eqref{eq:P_pair} that the pairwise confusion rate is determined by the Euclidean distance $D(\mathbf{x},\mathbf{x'})=\|\mathbf{x}' - \mathbf{x}\|$ between
codewords. We therefore proceed to establish a bound on $\|\mathbf{x}' - \mathbf{x}\|$.

\subsection{Typical Distance between Modulated Codewords}
The minimum Euclidean distance $D_{\min}$ is fundamentally related to the minimum Hamming distance $d_{H,\min}$ imposed by the Hamming bound, which motivates us to find the relationship between the two quantities. \revise{Expanding the squared Euclidean distance gives $D(\mathbf{x},\mathbf{x}')^2=\sum_{i:x_i\ne x_i'}\big(|x_i|^2+|x_i'|^2-2|x_i|\,|x_i'|\cos(\angle x_i-\angle x_i')\big)$,
where $\angle x_i$ is the phase of the complex symbol $x_i$ and $2|x_i||x_i'|\cos(\angle x_i-\angle x_i')=2\,\mathrm{Re}(x_ix_i'^*)$.
The Euclidean distance is thus determined not only by how many symbols differ, but also by their magnitudes and phase differences.
Thus, $D$ is not a function of $d_H(\mathbf{x},\mathbf{x}')$ alone and admits no closed form. We therefore characterize it in expectation, where $\bar{E}_s=\mathbb{E}[\vert x_i\vert^2]$ is the average symbol energy and $\bar{\eta}\triangleq\mathbb{E}[\text{Re}(x_ix_i'^*)\mid x_i\ne x_i']/\bar{E}_s$ is the normalized average cross-correlation between two distinct symbols, evaluated under the symbol-pair distribution induced by the codebook. Since $\bar{\eta}$ and $\Delta$ depend only on the constellation and its symbol-usage distribution, the relation holds for any codebook, once these are fixed, $\Delta$ is a constant and the analysis applies unchanged. We emphasize that the relation is an \emph{ensemble} identity rather than a per-pair equality: for a fixed codeword pair, $D$ also depends on the symbol phases and thus fluctuates around its conditional mean.}{For a codeword pair with $d=d_H(\mathbf{c},\mathbf{c}')$, let $\delta_i\triangleq|x_i-x_i'|^2$ denote the squared distance at a position where the two codewords differ. Then, exactly,}
\begin{align}
\label{eq:D2_exact}
    D^2(\mathbf{x},\mathbf{x}')=d\,\bar{\delta}(\mathbf{x},\mathbf{x}'),\quad
    \bar{\delta}(\mathbf{x},\mathbf{x}')\triangleq\frac{1}{d}\sum_{i:x_i\neq x_i'}\delta_i,
\end{align}
\revise{}{where $\bar{\delta}$ is the empirical mean of $\delta_i$ over the mismatched positions of this specific pair. The Euclidean distance is thus determined not only by how many symbols differ, but also by which symbol pairs occur at the mismatched positions, so $D$ is not a function of $d_H(\mathbf{c},\mathbf{c}')$ alone. A mismatched position carries an ordered symbol pair $(a,b)\in\mathcal{M}^2$, $a\neq b$, and there are $M(M-1)$ such pairs. With the symbols used uniformly, the mean of $|a-b|^2$ over them defines the per-symbol squared-distance coefficient}
\begin{align}
\label{eq:Delta}
    \Delta \triangleq \frac{1}{M(M-1)}\sum_{a\neq b\in\mathcal{M}}|a-b|^2 = 2\bar{E}_s\,(1-\bar{\eta}),
\end{align}
\revise{}{where $\bar{E}_s=\tfrac1M\sum_{x\in\mathcal{M}}\vert x\vert^2$ is the average symbol energy and $\bar{\eta}=-1/(M-1)$ is the normalized average cross-correlation between two distinct symbols. We then define the \emph{typical distance} at Hamming distance $d$ as}
\begin{align}
\label{eq:eff_distance}
    D_{\mathrm{typ}}(d) \triangleq \sqrt{\Delta\, d},
\end{align}
\revise{}{which serves as the representative distance in the bounds below. If, for a tolerance $\tau\geqslant0$, the empirical mean of a codeword pair satisfies $|\bar{\delta}(\mathbf{x},\mathbf{x}')-\Delta|\leqslant\tau\Delta$, i.e., the symbol pairs at its mismatched positions are mixed closely enough to the constellation average, then~\eqref{eq:D2_exact} yields}
\begin{align}
\label{eq:D2_sandwich}
    (1-\tau)\,\Delta\, d\leqslant D^2(\mathbf{x},\mathbf{x}')\leqslant(1+\tau)\,\Delta\, d.
\end{align}
\revise{}{This condition is weaker than strong typicality of the joint type of the symbol pairs and only requires the mean of $|a-b|^2$ to be close to $\Delta$. Since, by~\eqref{eq:D2_exact}, $D^2$ depends on the symbol pairs only through this mean, the weaker condition is sufficient for our analysis. For random codes, where the codeword pair and hence $\bar{\delta}$ are random, the following lemma bounds the probability that the condition is violated.}
\begin{lemma}
\label{lem:conc}
\revise{}{Let $\mathbf{c},\mathbf{c}'$ be drawn independently and uniformly from $\mathcal{Q}^n$, i.e., with i.i.d.\ uniform symbols over $\mathcal{Q}$, conditioned on $d_H(\mathbf{c},\mathbf{c}')=d$, and let $\mathbf{x}=\varphi(\mathbf{c})$ and $\mathbf{x}'=\varphi(\mathbf{c}')$. Then the $\delta_i$ at the mismatched positions are i.i.d.\ with mean $\Delta$, variance $\sigma_\delta^2$, and $|\delta_i-\Delta|\leqslant\kappa\triangleq d_{s,\max}^2-\Delta$, where $d_{s,\max}^2\triangleq\max_{a\neq b\in\mathcal{M}}|a-b|^2$, and}
\begin{align}
\label{eq:bernstein}
    P\!\left(|\bar{\delta}(\mathbf{x},\mathbf{x}')-\Delta|>\tau\Delta\right)\leqslant 2\exp\!\left(-\frac{d\,\tau^2\Delta^2}{2\sigma_\delta^2+\frac{2}{3}\kappa\,\tau\Delta}\right).
\end{align}
\end{lemma}
\begin{proof}
Conditioned on the mismatched positions, each of them carries an ordered symbol pair drawn uniformly and independently from the $M(M-1)$ pairs, so the $\delta_i$ are i.i.d.\ with mean $\Delta$ and $|\delta_i-\Delta|\leqslant\kappa$; \eqref{eq:bernstein} is Bernstein's inequality.
\end{proof}
\revise{}{Compared with \ac{psk}, the per-position squared distances of \ac{qam} spread more widely, e.g., $\sigma_\delta/\Delta\approx0.75$ for 16-\ac{qam} versus $0.35$ for \ac{qpsk}, so longer Hamming distances are needed for concentration: for 16-\ac{qam} at $\tau=1/3$, the right-hand side of~\eqref{eq:bernstein} equals $0.23$, $0.026$, and $3.4\times10^{-4}$ at $d=32$, $64$, and $128$.}
\begin{remark}[Typical-distance assumption]
\label{rem:uniform}
\revise{}{In the bounds below, we assume $\tau=0$, i.e., $\bar{\delta}(\mathbf{x},\mathbf{x}')=\Delta$, 
so that $D(\mathbf{x},\mathbf{x}')=D_{\mathrm{typ}}(d_H(\mathbf{c},\mathbf{c}'))=\sqrt{\Delta\,d_H(\mathbf{c},\mathbf{c}')}$ for every pair. 
This is an assumption on the code rather than an identity: it is exact for BPSK, whose distinct symbols are equidistant, 
and holds approximately for random codes at large Hamming distances by Lemma~\ref{lem:conc}. When a nonzero $\tau$ is known, 
the bounds hold with $\Delta$ replaced by the endpoint of~\eqref{eq:D2_sandwich} that keeps the respective inequality valid. 
For short structured codes, $\tau$ must be determined from the code.}
\end{remark}

\subsection{Bounds of Minimum Hamming Distance}

Since $P_{\mathrm{pair}}$ decreases with the Euclidean distance, which is tied to the Hamming distance via $\Delta$ in~\revise{(13)}{\eqref{eq:eff_distance}}, bounding $P_{\mathrm{con}}$ reduces to \revise{bounding the minimum Hamming distance $d_{H,\min}$. We therefore first bound $d_{H,\min}$ between a feasibility floor from the non-overlapping decision regions and the Hamming bound:}{characterizing the minimum Hamming distance $d_{H,\min}$. The Hamming, Plotkin, and Singleton bounds limit $d_{H,\min}$ from above, while the non-overlapping decision regions, evaluated at the typical distance~\eqref{eq:eff_distance}, give a feasibility threshold from below:}
\begin{lemma}
\label{lem:dmax}
    Given fixed $(M,n,k)$ and an \revise{$\varepsilon$-bounded}{error-bounded} decoder with decision radius $R(\varepsilon)$, \revise{the minimum Hamming distance $d_{H,\min}$ of a codebook $\mathcal{X}$ of size $M^k$ satisfies $d_{H,\min}^{\mathrm{LB}}\leqslant d_{H,\min}\leqslant d_{H,\min}^{\mathrm{UB}}$, where $D_{s,\max}^2\triangleq\max_{x\neq x'\in\mathcal{M}}|x-x'|^2$ is the maximum squared symbol distance,}{let}
    \begin{align}
    \label{eq:d_min^min}
        d_{H,\min}^{\mathrm{LB}} &= \left\lceil\frac{4R^2(\varepsilon)}{\Delta}\right\rceil, \\
    \label{eq:d_min^max}
        d_{H,\min}^{\mathrm{UB}} &= \min\bigg\{d_{H}^{\mathrm{Ham}},\ \bigg\lfloor\frac{n(M-1)M^{k-1}}{M^k-1}\bigg\rfloor,\ n-k+1\bigg\},
    \end{align}
    \revise{}{with the Hamming-bound term}
    \begin{align}
    \label{eq:d_Ham}
        d_{H}^{\mathrm{Ham}} = 2\min\bigg\{T\in\mathbb{N}\,\Big|\sum_{j=0}^T\tbinom{n}{j}(M-1)^j>M^{r}\bigg\}.
    \end{align}
    \revise{}{Then the minimum Hamming distance of any codebook $\mathcal{X}$ of size $M^k$ satisfies $d_{H,\min}\leqslant d_{H,\min}^{\mathrm{UB}}$, and a Hamming shell $d$ satisfies the non-overlapping condition at its typical distance, $D_{\mathrm{typ}}(d)\geqslant 2R(\varepsilon)$, if and only if $d\geqslant d_{H,\min}^{\mathrm{LB}}$.}
\end{lemma}
\begin{proof}
The Hamming-bound term $d_{H}^{\mathrm{Ham}}$ follows as in~\cite{han2025confusions} Lemma 1, while the second and third terms of~\eqref{eq:d_min^max} are the Plotkin and Singleton bounds. Each of them upper-bounds the minimum Hamming distance of any codebook of size $M^k$, and so does their minimum.
By~\eqref{eq:eff_distance}, $D_{\mathrm{typ}}^2(d)=\Delta d$, so $D_{\mathrm{typ}}(d)\geqslant 2R(\varepsilon)$ holds if and only if $d\geqslant 4R^2(\varepsilon)/\Delta$. 
Since $d$ is an integer, this is $d\geqslant d_{H,\min}^{\mathrm{LB}}$.
\end{proof}

\subsection{Bound of Block Confusion Rate}
\revise{Then we can derive the monotonicity of $P_\text{pair}$ w.r.t $D$:}{By~\eqref{eq:P_con}, $P_{\mathrm{con}}$ is a sum of pairwise confusion rates $P_{\mathrm{pair}}(D)$, so bounds on the Euclidean distances between codewords translate into bounds on $P_{\mathrm{con}}$ once $P_{\mathrm{pair}}$ is known to be monotonic in $D$. We therefore first establish the monotonicity of $P_{\mathrm{pair}}$:}

\begin{lemma}
\label{lem:lemma2}
    Given fixed $(\varepsilon, \sigma)$, $P_{\mathrm{pair}}$
is monotonically decreasing and convex in $D$ for $D \geqslant 2R(\varepsilon)$.
\end{lemma}
The proof is identical to that of~\cite{han2025confusions} Lemma 2. This yields a lower bound on $P_{\mathrm{con}}$:

\begin{theorem}
\label{theo:theorem1}
For fixed $(M, n, k, \bar{E}_s)$, consider any \revise{$\varepsilon$-bounded decoder that rejects list outputs}{error-bounded decoder operating in the unique-decoding regime}. If the codewords of \revise{$\mathcal{X}$ are uniformly distributed over $\mathcal{M}^n$}{$\mathcal{C}$ are uniformly distributed over $\mathcal{Q}^n$}\revise{}{ and satisfy the typical-distance assumption of Remark~\ref{rem:uniform}}, then \revise{the block confusion rate $P_{\mathrm{con}}$ satisfies:}{the ensemble-average block confusion rate $\bar{P}_{\mathrm{con}}$, i.e., $P_{\mathrm{con}}$ averaged over the random choice of the codebook, satisfies $\bar{P}_{\mathrm{con}}\geqslant \bar{P}_{\mathrm{con}}^{\mathrm{LB}}$, where}
\begin{equation}
\label{eq:P_con_LB}
    \bar{P}_{\mathrm{con}}^{\mathrm{LB}} \! \triangleq \!
    \frac{M^k-1}{M^{n}}
    \binom{n}{d_{H,\mathrm{min}}^{\mathrm{UB}}}
    (M-1)^{d_{H,\mathrm{min}}^{\mathrm{UB}}}
    P_{\mathrm{pair}}\!\left(\sqrt{\Delta \cdot d_{H,\mathrm{min}}^{\mathrm{UB}}}\right).
\end{equation}
\end{theorem}
The proof follows \cite{han2025confusions} with $E$ replaced by $\Delta$. \revise{We lower-bound the confusion rate by keeping only the shell at the maximum achievable minimum distance $d_{H,\min}^{\mathrm{UB}}$---the farthest, least-confusable shell---and dropping the others. Its population $A_d=\binom{n}{d}(M-1)^{d}/M^{n-k}$, the expected number of codewords at Hamming distance $d$, is obtained by assuming the codewords are uniformly distributed over $\mathcal{M}^n$.}{Under the typical-distance assumption, every codeword pair at Hamming distance $d$ is at distance $\sqrt{\Delta d}$, so grouping the terms of~\eqref{eq:P_con} by Hamming distance gives $P_{\mathrm{con}}=\sum_{d=1}^{n}A_d(\mathcal{C})\,P_{\mathrm{pair}}(\sqrt{\Delta d})$, where $A_d(\mathcal{C})\triangleq\frac{1}{|\mathcal{X}|}\sum_{\mathbf{x}\in\mathcal{X}}|\mathcal{V}(\mathcal{X},\mathbf{x},d)|$ is the average number of codewords at Hamming distance $d$ in the given codebook. When the $M^k$ codewords are drawn independently and uniformly from $\mathcal{Q}^n$, each of the other $M^k-1$ codewords lies at Hamming distance $d$ from a given codeword with probability $\binom{n}{d}(M-1)^d/M^n$. Averaging over the codebook therefore yields $\bar{P}_{\mathrm{con}}=\sum_{d=1}^{n}A_d\,P_{\mathrm{pair}}(\sqrt{\Delta d})$, where $A_d=(M^k-1)\binom{n}{d}(M-1)^{d}/M^{n}$ is the average of $A_d(\mathcal{C})$ over the codebook, an ensemble quantity rather than a property of an individual code. Since every term of this sum is non-negative, keeping a single shell lower-bounds $\bar{P}_{\mathrm{con}}$. We keep the shell at $d_{H,\min}^{\mathrm{UB}}$, the upper bound on the minimum Hamming distance of any codebook of size $M^k$, which yields~\eqref{eq:P_con_LB} as a conservative bound.} \revise{Within this shell, Lemma~3 further replaces $\mathbb{E}[P_{\mathrm{pair}}(D)\mid d]$ by its smaller \ac{rms}-distance value $P_{\mathrm{pair}}(\sqrt{\Delta\,d_{H,\min}^{\mathrm{UB}}})$.}{}


\begin{theorem}
\label{theo:theorem2}
For fixed $(M, k, \bar{E}_s)$, \revise{and assuming the $M^k$ codewords are uniformly distributed over the $M$-ary space $\mathcal{M}^n$, any $\varepsilon$-bounded decoder that rejects list-outputs satisfies:}{consider an error-bounded decoder and a codebook $\mathcal{C}\subset\mathcal{Q}^n$ of size $M^k$ whose modulated codewords have a minimum Euclidean distance that reaches the typical threshold distance, $D_{\min}\geqslant\sqrt{\Delta\,d_{H,\min}^{\mathrm{LB}}}$. Then $P_{\mathrm{con}}\leqslant P_{\mathrm{con}}^{\mathrm{UB}}$, where}
\begin{equation}
\label{eq:P_con^UB}
    P_{\mathrm{con}}^{\mathrm{UB}} \triangleq
    \left(M^{k}-1\right)P_{\mathrm{pair}}\!\left(\sqrt{\Delta\cdot d_{H,\min}^{\mathrm{LB}}}\right).
\end{equation}
\end{theorem}
\begin{proof}
By Lemma~\ref{lem:dmax}, $D_{\min}\geqslant\sqrt{\Delta\,d_{H,\min}^{\mathrm{LB}}}\geqslant 2R(\varepsilon)$, 
so the decision regions do not overlap and~\eqref{eq:P_con} holds with equality. For every transmitted $\mathbf{x}$, 
the inner sum in~\eqref{eq:P_con} has $|\mathcal{X}|-1=M^k-1$ terms, and each satisfies 
$P(\mathbf{x}\to\mathbf{x}')=P_{\mathrm{pair}}(D(\mathbf{x},\mathbf{x}'))\leqslant P_{\mathrm{pair}}(\sqrt{\Delta\,d_{H,\min}^{\mathrm{LB}}})$ 
since $P_{\mathrm{pair}}$ is decreasing for $D\geqslant 2R(\varepsilon)$ (Lemma~\ref{lem:lemma2}). 
Averaging over $\mathbf{x}$ gives $P_{\mathrm{con}}\leqslant(M^k-1)P_{\mathrm{pair}}(\sqrt{\Delta\,d_{H,\min}^{\mathrm{LB}}})$.
\end{proof}

\section{Sensitivity of Confusion Rate Bounds} \label{sec:sensitivity}
In this section, we examine the sensitivity of the confusion rate bounds 
$\bar{P}_{\mathrm{con}}^{\mathrm{LB}}$ and $P_{\mathrm{con}}^{\mathrm{UB}}$ 
to variations in the energy and blocklength.

\subsection{Lower Bound versus Energy}
\begin{theorem}
\label{theo:P_con_LB_E_0}
For fixed $(M, k, n, \sigma, \varepsilon)$, the lower bound $\bar{P}_{\mathrm{con}}^{\mathrm{LB}}$ is monotonically decreasing and convex in the average symbol energy $\bar{E}_s$.
\end{theorem}
\begin{proof}
With $\alpha\triangleq\Delta/\bar{E}_s=2(1-\bar{\eta})$ and $C$ the prefactor in~\eqref{eq:P_con_LB}, $\bar{P}_{\mathrm{con}}^{\mathrm{LB}}=C\,g(h(\bar{E}_s))$, where $g\triangleq P_{\mathrm{pair}}$ and $h(\bar{E}_s)\triangleq\sqrt{\alpha\,d_{H,\min}^{\mathrm{UB}}\bar{E}_s}$. For $d_{H,\min}^{\mathrm{UB}}\geqslant d_{H,\min}^{\mathrm{LB}}$, $h\geqslant 2R(\varepsilon)$, so $g'\leqslant0$ and $g''\geqslant0$ by Lemma~\ref{lem:lemma2}, while $h'>0$ and $h''<0$. Hence $(g\circ h)'=g'h'\leqslant0$ and $(g\circ h)''=g''h'^2+g'h''\geqslant0$.
\end{proof}

\subsection{Lower Bound versus Blocklength}
\begin{lemma}
\label{lemm:d_min^max_n}
For fixed $(M, k)$, $d_{H,\min}^{\mathrm{UB}}$ is \revise{monotonically increasing}{nondecreasing} in $n$.
\end{lemma}
\begin{proof}
The Hamming-bound term $d_{H}^{\mathrm{Ham}}$ is nondecreasing in $n$, as shown in~\cite{han2025confusions}, and the Plotkin and Singleton terms of~\eqref{eq:d_min^max} increase linearly in $n$. Hence their minimum $d_{H,\min}^{\mathrm{UB}}$ is nondecreasing in $n$.
\end{proof}

\begin{lemma}
\label{lem:d_min^max(n)}
For all $n \in [k, M^k - 2]$, \revise{the maximum achievable minimum Hamming
distance}{the upper bound $d_{H,\min}^{\mathrm{UB}}$ in~\eqref{eq:d_min^max}} satisfies $d_{H,\min}^{\mathrm{UB}}(n) < \frac{(M-1)(n+1)}{M}$.
\end{lemma}

\begin{proof}
By~\eqref{eq:d_min^max}, $d_{H,\min}^{\mathrm{UB}}(n)$ does not exceed its Plotkin term, and $\frac{n(M-1)M^{k-1}}{M^k-1}<\frac{(n+1)(M-1)M^{k-1}}{M^k}=\frac{(M-1)(n+1)}{M}$, since $\frac{n}{M^k-1}<\frac{n+1}{M^k}$ is equivalent to $n<M^k-1$.
\end{proof}

\begin{theorem}
\label{theo:P_con^LB_n}
For fixed $(M, k, \bar{E}_s, \sigma, \varepsilon)$, $\bar{P}_{\mathrm{con}}^{\mathrm{LB}}$ is piecewise monotonically decreasing
in $n$ within each $d_{H,\min}^{\mathrm{UB}}$-consistent interval, i.e. a plateau of the nondecreasing integer 
step function $d_{H,\min}^{\mathrm{UB}}(n)$, and globally monotonically decreasing for sufficiently large $n$.
\begin{proof}
Compared with \ac{psk} constellations \cite{han2025confusions}, the \ac{qam} model only replaces the per-shell distance $\sqrt{E\,d}$ by $\sqrt{\Delta\,d}$ and the noise dimension $n$ by $2n$, both constant in $n$, so the piecewise and eventual global monotonicity in $n$ carry over unchanged.
\end{proof}
\end{theorem}
\revise{}{When $d_{H,\min}^{\mathrm{UB}}$ is determined by its Singleton term, as in the setting of Fig.~\ref{fig:P_con_n}, every such interval contains a single $n$, and $\bar{P}_{\mathrm{con}}^{\mathrm{LB}}$ decreases at every step. Plateaus, and hence the piecewise behavior, arise only where the Hamming or Plotkin term is the smallest, e.g., at low code rates.}
\subsection{Upper Bound versus Energy}
\begin{lemma}
For any fixed $(\sigma, n, M, \varepsilon)$, $d_{H,\min}^{\mathrm{LB}}$ is \revise{monotonically decreasing}{non-increasing} in $\bar{E}_s$.
\end{lemma}
\begin{proof}
From~\eqref{eq:d_min^min}, $d_{H,\min}^{\mathrm{LB}}=\lceil 4R^2(\varepsilon)/\Delta\rceil$. Since
$\bar{\eta}$ is invariant to scaling the constellation, $\Delta=2\bar{E}_s(1-\bar{\eta})$ is strictly increasing in $\bar{E}_s$, while $R(\varepsilon)$ is
independent of $\bar{E}_s$. Hence the $4R^2(\varepsilon)/\Delta$ is strictly decreasing
in $\bar{E}_s$, and $d_{H,\min}^{\mathrm{LB}}$ is non-increasing.
\end{proof}

\begin{theorem}
\label{theo:E_0,i}
For fixed $(M, k, n, \sigma, \varepsilon)$, the upper bound 
$P_{\mathrm{con}}^{\mathrm{UB}}$ is piecewise continuous in $\bar{E}_s$, with jump discontinuities at
    $\bar{E}_{s,i} = \frac{2R^2(\varepsilon)}{(1-\bar{\eta})\,i}$, $i \in \mathbb{N}^+$ and is monotonically decreasing and convex within each continuity interval.
\end{theorem}

\begin{proof}
Substituting $\Delta=2\bar{E}_s(1-\bar{\eta})$ into~\eqref{eq:d_min^min} gives $d_{H,\min}^{\mathrm{LB}}=\lceil 2R^2(\varepsilon)/((1-\bar{\eta})\bar{E}_s)\rceil$, a staircase function of $\bar{E}_s$ that jumps where the argument of the ceiling equals an integer $i$, i.e., at $\bar{E}_{s,i}$. Within each interval, $d_{H,\min}^{\mathrm{LB}}$ is constant and $\sqrt{\Delta\,d_{H,\min}^{\mathrm{LB}}}\geqslant 2R(\varepsilon)$ by~\eqref{eq:d_min^min}, so the proof of Theorem~\ref{theo:P_con_LB_E_0} applies with $C=M^k-1$ and $d_{H,\min}^{\mathrm{UB}}$ replaced by $d_{H,\min}^{\mathrm{LB}}$. At $\bar{E}_{s,i}$, the argument returns to $\sqrt{\Delta\,i}=2R(\varepsilon)$, i.e., $P_{\mathrm{con}}^{\mathrm{UB}}$ jumps up to the envelope $(M^k-1)P_{\mathrm{pair}}(2R(\varepsilon))$.
\end{proof}

\subsection{Upper Bound versus Blocklength}
\begin{lemma}
\label{lem:d_min^min_n}
For fixed $(M, k, \bar{E}_s, \sigma, \varepsilon)$, $d_{H,\min}^{\mathrm{LB}}$ is \revise{monotonically increasing}{nondecreasing} in $n$.
\end{lemma}
\begin{proof}
From~\eqref{eq:d_min^min}, $d_{H,\min}^{\mathrm{LB}}=\lceil 4R^2(\varepsilon)/\Delta\rceil$. The decision
radius $R(\varepsilon)=\sigma F^{-1}_{\chi_{2n}}(1-\varepsilon)$ is increasing in $n$, since the
$(1-\varepsilon)$ quantile of $\chi_{2n}$ grows with its degrees of freedom, while $\Delta$ is
independent of $n$. Hence the argument $4R^2(\varepsilon)/\Delta$ is increasing in $n$, and $d_{H,\min}^{\mathrm{LB}}$ is non-decreasing.
\end{proof}

\begin{lemma}
\label{lem:Ppair_n}
\revise{}{For fixed $(D,\sigma,\varepsilon)$, the pairwise confusion rate $P_{\mathrm{pair}}(D)$, evaluated with the decision radius $R(\varepsilon)$ and the noise dimension $2n$ of blocklength $n$, strictly increases with $n$.}
\end{lemma}
\begin{proof}
See Appendix~\ref{appendix:Ppair_n}.
\end{proof}

\begin{theorem}
\label{theo:P_con^UB_n}
\revise{For fixed $(M,k,\bar{E}_s,\sigma,\varepsilon)$, the upper bound
$P_{\mathrm{con}}^{\mathrm{UB}}$ is monotonically decreasing in $n$ for
sufficiently large $n$.}{For fixed $(M,k,\bar{E}_s,\sigma,\varepsilon)$, the upper bound $P_{\mathrm{con}}^{\mathrm{UB}}$ strictly increases with $n$ within every interval of $n$ on which $d_{H,\min}^{\mathrm{LB}}$ is constant.}
\end{theorem}

\begin{proof}
Within an interval of constant $d_{H,\min}^{\mathrm{LB}}$, the argument $\sqrt{\Delta\,d_{H,\min}^{\mathrm{LB}}}$ of $P_{\mathrm{pair}}$ in~\eqref{eq:P_con^UB} is fixed, since $\Delta$ does not depend on $n$, so $P_{\mathrm{con}}^{\mathrm{UB}}$ strictly increases with $n$ by Lemma~\ref{lem:Ppair_n}.
\end{proof}
\revise{}{Since $d_{H,\min}^{\mathrm{LB}}$ is a nondecreasing step function of 
$n$ (Lemma~\ref{lem:d_min^min_n}), $P_{\mathrm{con}}^{\mathrm{UB}}$ rises within each 
interval and drops at every increment of $d_{H,\min}^{\mathrm{LB}}$, where the typical distance $D_{\mathrm{typ}}(d_{H,\min}^{\mathrm{LB}})=\sqrt{\Delta\,d_{H,\min}^{\mathrm{LB}}}$ jumps from $\sqrt{\Delta\,i}$ to $\sqrt{\Delta\,(i+1)}$. The result is a sawtooth pattern whose local maxima occur at 
the right boundaries of the intervals, where $\sqrt{\Delta\,d_{H,\min}^{\mathrm{LB}}}$ 
is closest to $2R(\varepsilon)$. These maxima are bounded by the envelope 
$(M^k-1)P_{\mathrm{pair}}(2R(\varepsilon))=(M^k-1)F_{\chi'_{2n}(2R(\varepsilon)/\sigma)}(R(\varepsilon)/\sigma)$,
where $F_{\chi'_{2n}(\mu)}$ is the \ac{cdf} of the noncentral $\chi$ distribution with $2n$ degrees of freedom and noncentrality parameter $\mu$. The envelope depends on $n$ only through the noise dimension $2n$, 
and we verified numerically that it decreases strictly in $n$ in Section \ref{sec:numerical}.}

\section{Numerical Experiment Results} \label{sec:numerical}
We set $\varepsilon = 0.05$ and $\sigma^2 = 0.5$ throughout all experiments.
The simulation is conducted for square $M$-\ac{qam} constellations with $M \in
\{16, 64\}$. Throughout, we do not apply probabilistic shaping and assume the
constellation symbols are used uniformly, so that $\bar{\eta}=-1/(M-1)$ and
$\Delta=2M\bar{E}_s/(M-1)$. As $\bar{\eta}$ and $\Delta$ are constants once the
codebook is fixed, a non-uniformly used codebook would only
rescale $\Delta$ through $\bar{\eta}$, leaving the analysis and the qualitative
behavior unchanged.
\begin{figure}[!htpb]
    \centering
    \includegraphics[width=0.40\textwidth]{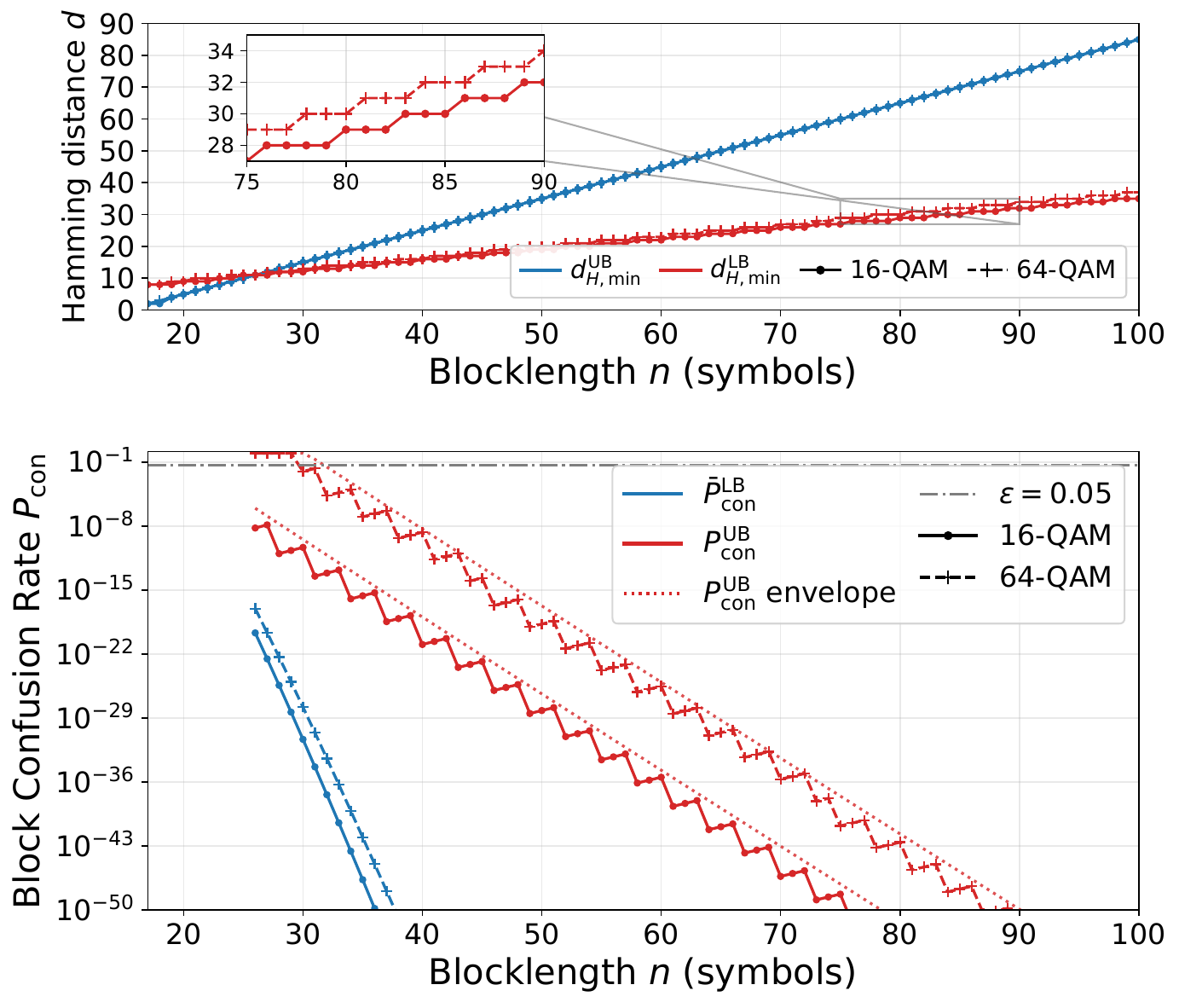}
    \caption{Bounds of block confusion rate versus blocklength for 16-QAM and 64-QAM with $k=16$ and $E_s/N_0=8$~dB.}
    \label{fig:P_con_n}
\end{figure}

Fig.~\ref{fig:P_con_n} plots the confusion-rate bounds $\bar{P}_{\mathrm{con}}^{\mathrm{LB}}$ and $P_{\mathrm{con}}^{\mathrm{UB}}$ versus 
blocklength $n$, for $k=16$ and $E_s/N_0=8$~dB, under 16- and 64-\ac{qam}. \revise{Each upper subplot shows}{The upper panel shows}
the limits of minimum Hamming-distance $d_{H,\min}^{\mathrm{LB}}$ and $d_{H,\min}^{\mathrm{UB}}$.
Both bounds increase monotonically in $n$: \revise{$d_{H,\min}^{\mathrm{UB}}$ increases with the codebook's expanding packing capacity}
{$d_{H,\min}^{\mathrm{UB}}$ is given in this range by its Singleton term $n-k+1$, so the curves of the two constellations overlap}, and $d_{H,\min}^{\mathrm{LB}}$
increases because $R(\varepsilon)$ grows with $n$. \revise{The lower bound $P_{\mathrm{con}}^{\mathrm{LB}}$
is piecewise decreasing within each $d_{H,\min}^{\mathrm{UB}}$-consistent interval since adding redundancy at fixed $k$ makes the codebook sparser; and
it decreases globally for large $n$. The upper bound $P_{\mathrm{con}}^{\mathrm{UB}}$ likewise
decreases for large $n$, as the rising floor $d_{H,\min}^{\mathrm{LB}}$ drops the nearest shells.}
{The lower bound $\bar{P}_{\mathrm{con}}^{\mathrm{LB}}$ decreases monotonically, 
since adding redundancy at fixed $k$ makes the codebook sparser. 
The upper bound $P_{\mathrm{con}}^{\mathrm{UB}}$ is not monotonic in $n$: 
within each interval of constant $d_{H,\min}^{\mathrm{LB}}$, it rises with $R(\varepsilon)$, 
and it drops whenever $d_{H,\min}^{\mathrm{LB}}$ increases. 
This sawtooth pattern, consistent with Theorem~\ref{theo:P_con^UB_n}, nevertheless decreases overall with $n$.}
Higher-order constellations need more redundancy to bring $P_{\mathrm{con}}^{\mathrm{UB}}$ below $\varepsilon$.\revise{}
{ For 64-\ac{qam} at short blocklengths, the prefactor $M^k-1$ even lifts $P_{\mathrm{con}}^{\mathrm{UB}}$ above $\varepsilon$, where the upper bound is not informative.}

Fig.~\ref{fig:P_con_snr} plots the same two bounds as functions of the \ac{snr} $E_s/N_0$, for $(n,k)=(32,16)$ under 16-\ac{qam} and 64-\ac{qam}. 
As before, the main plot shows the two bounds and the upper \revise{subplot}{panel} shows $d_{H,\min}^{\mathrm{LB}}$ and $d_{H,\min}^{\mathrm{UB}}$, \revise{delimiting the feasible
and infeasible ranges}{whose crossing marks where both bounds become defined}. As the \ac{snr} grows, $d_{H,\min}^{\mathrm{UB}}$ remains constant since it depends only on $(n,k,M)$, while 
$d_{H,\min}^{\mathrm{LB}}$ decreases. 
Consistent with the analysis, $\bar{P}_{\mathrm{con}}^{\mathrm{LB}}$ is monotonically decreasing and
convex in \ac{snr}, whereas $P_{\mathrm{con}}^{\mathrm{UB}}$ is piecewise continuous with jump discontinuities at $\bar{E}_{s,i}$, 
decreasing within each interval. \revise{ and each jump marks an \ac{snr} at which a newly feasible nearest shell enters the sum.}
{Each jump marks an \ac{snr} at which $d_{H,\min}^{\mathrm{LB}}$ decreases by one, so that the typical distance $\sqrt{\Delta\,d_{H,\min}^{\mathrm{LB}}}$ falls back to $2R(\varepsilon)$ and $P_{\mathrm{con}}^{\mathrm{UB}}$ returns to the envelope.} Higher \ac{snr} pushes codewords
apart from the fixed noise ball, so the pairwise confusion and both bounds fall.
\revise{}{Comparing the two constellations, 64-\ac{qam} yields higher values of both bounds.
At the same average symbol energy, its points are packed more densely, so a codeword pair differing in a given number of symbols lies closer in Euclidean distance,
which raises $d_{H,\min}^{\mathrm{LB}}$ and the pairwise confusion rate at a given Hamming distance.
Moreover, the larger alphabet gives many more competing codewords for the same payload, which lifts $P_{\mathrm{con}}^{\mathrm{UB}}$ further.}

\begin{figure}[!htpb]
    \centering
    \includegraphics[width=0.40\textwidth]{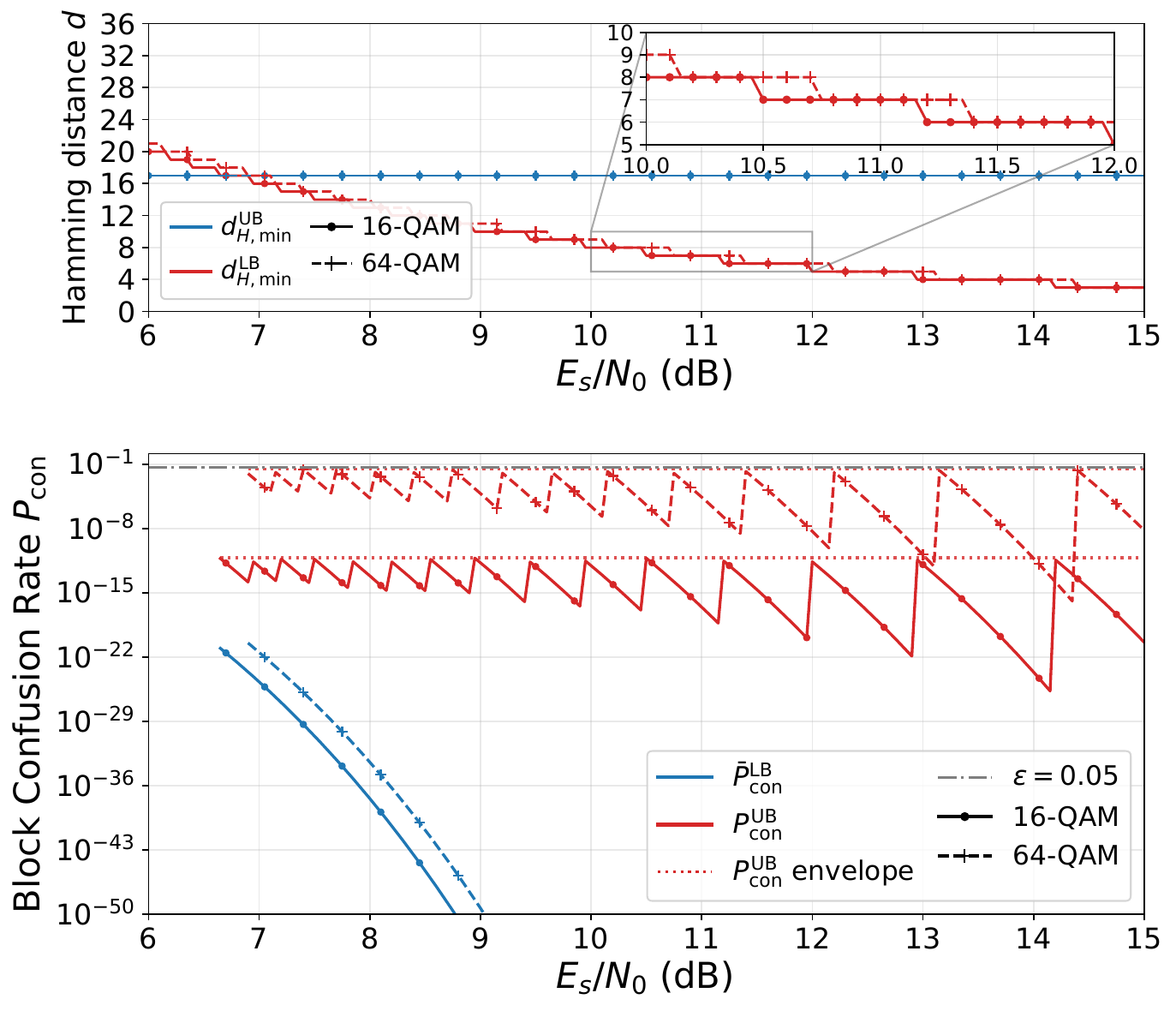}
    \caption{Bounds of block confusion rate versus $E_s/N_0$ for 16-QAM and 64-QAM with $(32,16)$ codebooks.}
    \label{fig:P_con_snr}
\end{figure}

\revise{}{To validate the analysis on practical codes, Fig.~\ref{fig:P_con_n8k4} 
considers $(n,k)=(8,4)$, for which the codebook is small enough to run the error-bounded decoder
exactly by comparing $\mathbf{y}$ with every modulated codeword. 
We use the binary polar and extended BCH (eBCH) codes of length 32 and payload 16, whose 32 coded bits are mapped onto $n=8$ Gray-labelled 16-\ac{qam} symbols, 
and a Reed--Solomon (RS) code of length 8 and payload 4 over $\mathrm{GF}(16)$, whose code symbols are mapped directly onto the 16-\ac{qam} points. 
The unique-decoding regimes of the three codes start at different \acp{snr}, where $D_{\min}=2R$.
The polar and eBCH codes have $D_{\min}^2=8d_s^2$, where $d_s$ is the minimum distance between two 16-\ac{qam} symbols,
and enter the regime at 12.16~dB.
The RS code has $D_{\min}^2=6d_s^2$ and therefore requires 1.25~dB more.
Besides, the RS code has the highest confusion rate, since $P_{\mathrm{pair}}$ decays rapidly with the Euclidean distance, and the RS code has the smallest $D_{\min}$.
Likewise, at the same $D_{\min}$, the eBCH code attains a confusion rate about 3.2 times lower than the polar code, since it has fewer codeword pairs at $D_{\min}$.
In addition, all three codes stay below the envelope $(M^k-1)P_{\mathrm{pair}}(2R)$,
but they exceed $P_{\mathrm{con}}^{\mathrm{UB}}$ between the start of their unique-decoding regimes and 13.92~dB.
The reason is that $P_{\mathrm{con}}^{\mathrm{UB}}$ evaluates the pairwise confusion rate at the typical distance $\sqrt{\Delta\,d_{H,\min}^{\mathrm{LB}}}$ of a Hamming shell,
whereas the Euclidean distance of a specific codeword pair fluctuates around this value, as noted after Lemma~\ref{lem:dmax}.
For 16-\ac{qam}, the downward tolerance of all three codes attains its maximum $1-d_s^2/\Delta$, since each contains codeword pairs whose mismatched symbols are all nearest neighbours. 
Replacing $\Delta$ by $(1-\tau)\Delta=d_s^2$ as in Remark~\ref{rem:uniform} therefore yields the nearest-neighbour bound, 
shown as the $\tau$-corrected $P_{\mathrm{con}}^{\mathrm{UB}}$ in Fig.~\ref{fig:P_con_n8k4}, 
which holds for every code in the unique-decoding regime.

Finally, we validate~\eqref{eq:P_con} by Monte-Carlo simulation of the polar code with the exact error-bounded decoder.
Since $P_{\mathrm{con}}$ decreases by about two orders of magnitude every 0.5~dB, the simulation is restricted to the beginning of the unique-decoding regime.
We therefore simulate two \acp{snr}, which yield $2.33\times10^{-8}$ at 12.16~dB and $2.05\times10^{-9}$ at 12.5~dB,
in agreement with the exact values $2.18\times10^{-8}$ and $1.76\times10^{-9}$, respectively.}

\begin{figure}[!htpb]
    \centering
    \includegraphics[width=0.40\textwidth]{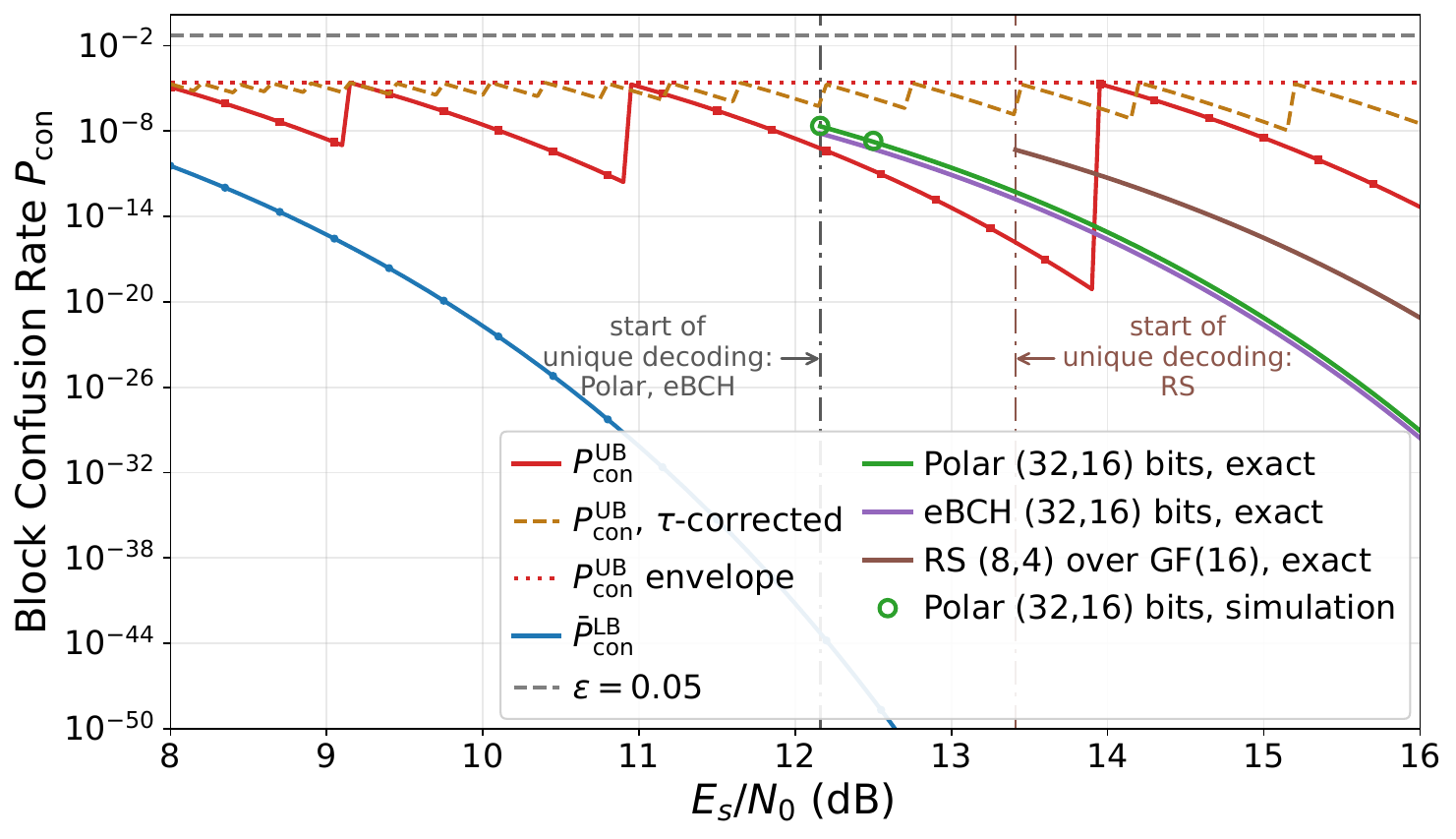}
    \caption{Bounds and confusion rates for $(n,k)=(8,4)$ 16-QAM symbols, $\varepsilon=0.05$.
    }
    \label{fig:P_con_n8k4}
\end{figure}

\section{Conclusion}\label{sec:conclusion} 
In this work, by introducing the per-symbol squared-distance coefficient,
analytical bounds on the block confusion rate were derived and shown
\revise{to be monotonically decreasing in both the average symbol energy and the blocklength, at a rate
governed by the constellation order.}
{to decrease in the average symbol energy and in the blocklength.
} 
Numerical results confirmed that as the \ac{snr} 
or redundancy increases, the confusion rate falls many orders of magnitude below the reliability 
target, \revise{leaving erasures as the dominant residual failure, which the code's inherent error detection
already flags, so a separate outer code such as \ac{crc} is no longer required.}{leaving erasures declared by the decoder as the dominant residual failure.} \revise{}{An exact computation for short polar, eBCH, and RS codes confirmed~\eqref{eq:P_con} in the unique-decoding regime and showed that codes whose minimum Euclidean distance falls below the typical threshold distance can exceed $P_{\mathrm{con}}^{\mathrm{UB}}$.}

For future work, two directions are of interest. First, \revise{incorporating the actual distance spectrum
of specific channel codes such as \ac{ldpc} or polar codes would tighten the bounds and tie them to deployed designs.}
{investigating the confusion rate of practical decoders, such as belief propagation for 
\ac{ldpc} codes and successive-cancellation list decoding for polar codes combined with a radius test, 
would tie the analysis to deployed systems. 
This also calls for extending the analysis beyond the unique-decoding regime.}
 Second, extending the framework to probabilistically shaped and non-square constellations would
broaden its applicability to practical modulation schemes.

\appendices
\linespread{1.0}\selectfont  

\section{Proof of Lemma \ref{lem:Ppair_n}}
\begin{proof}
\label{appendix:Ppair_n}
Let $\mathbf{v}=\mathbf{x}'-\mathbf{x}$ with $\|\mathbf{v}\|=D$, and normalize all quantities by $\sigma$. At blocklength $n$, $P_{\mathrm{pair}}(D)=P(\|\mathbf{w}-\mathbf{v}\|\leqslant R(\varepsilon))=F_T(a)$, where $T$ is noncentral chi-squared with $2n$ degrees of freedom and noncentrality $\lambda=D^2/\sigma^2$, and $a=R^2(\varepsilon)/\sigma^2$ at blocklength $n$. At blocklength $n+1$, the two codewords agree in the added symbol, so $D$ is unchanged and the added noise contributes an independent $Y\sim\chi^2_2$, giving $P_{\mathrm{pair}}(D)=\mathbb{E}_Y[F_T(b-Y)]$, where $b=R^2(\varepsilon)/\sigma^2$ at blocklength $n+1$. With $U\sim\chi^2_{2n}$, the calibration of the decision radius at both blocklengths gives $\mathbb{E}_Y[F_U(b-Y)]=P(U+Y\leqslant b)=1-\varepsilon=F_U(a)$.
The density of $T$ is a Poisson mixture of $\chi^2_{2n+2j}$ densities, $j\geqslant0$, and each ratio $f_{\chi^2_{2n+2j}}(t)/f_{\chi^2_{2n}}(t)\propto t^{j}$ is nondecreasing in $t$, strictly for $j\geqslant1$. Hence $f_T/f_U$ is strictly increasing for $\lambda>0$, and $\phi\triangleq F_T\circ F_U^{-1}$ is strictly convex on $(0,1)$, since $\phi'(u)=(f_T/f_U)(F_U^{-1}(u))$ is strictly increasing. Since $F_T=\phi\circ F_U$, Jensen's inequality applied to the non-degenerate random variable $F_U(b-Y)$ yields
\begin{align}
\begin{split}
\mathbb{E}_Y\big[F_T(b-Y)\big]&=\mathbb{E}_Y\big[\phi\big(F_U(b-Y)\big)\big]\\
&>\phi\big(\mathbb{E}_Y[F_U(b-Y)]\big)\\
&=\phi\big(F_U(a)\big)=F_T(a),
\end{split}
\end{align}
i.e., $P_{\mathrm{pair}}(D)$ at blocklength $n+1$ strictly exceeds its value at blocklength $n$.
\end{proof}

\bibliographystyle{IEEEtran}
\bibliography{references}

\end{document}

%% file: glossary.tex
\makeglossaries
\newacronym{awgn}{AWGN}{additive white Gaussian noise}
\newacronym{bec}{BEC}{binary erasure channel}
\newacronym{bler}{BLER}{block error rate}
\newacronym{cdf}{CDF}{cumulative distribution function}
\newacronym{cp}{CP}{control plane}
\newacronym{csi}{CSI}{channel state information}
\newacronym{csit}{CSIT}{channel state information at transmitter}
\newacronym{crc}{CRC}{cyclic redundancy check}
\newacronym{dft-s-ofdm}{DFT-s-OFDM}{Discrete Fourier Transform-spread-OFDM}
\newacronym{fbl}{FBL}{finite blocklength}
\newacronym{harq}{HARQ}{hybrid automatic repeat request}
\newacronym{gan}{GAN}{generative adversarial network}
\newacronym{ibl}{IBL}{infinite blocklength}
\newacronym{isac}{ISAC}{integrated sensing and communication}
\newacronym{xr}{XR}{extended reality}
\newacronym{ldpc}{LDPC}{low-density parity-check}
\newacronym{lfp}{LFP}{leakage-failure probability}
\newacronym{ml}{ML}{maximum likelihood}
\newacronym{map}{MAP}{maximum a posteriori}
\newacronym{mimo}{MIMO}{multi-input multi-output}
\newacronym{mm}{MM}{Majorize-Minimization}
\newacronym{noma}{NOMA}{non-orthogonal multi-access}
\newacronym{nom}{NOM}{non-orthogonal multiplexing}
\newacronym{ofdm}{OFDM}{orthogonal frequency-division multiplexing}
\newacronym{ofdma}{OFDMA}{orthogonal frequency-division multiple access}
\newacronym{qam}{QAM}{quadrature amplitude modulation}
\newacronym{oma}{OMA}{orthogonal multiple access}
\newacronym{papr}{PAPR}{Peak-to-Average Power Ratio}
\newacronym{per}{PER}{packet error rate}
\newacronym{phy}{PHY}{physical}
\newacronym{pls}{PLS}{physical layer security}
\newacronym{prb}{PRB}{physical resource block}
\newacronym{psk}{PSK}{phase-shift keying}
\newacronym{qpsk}{QPSK}{quadrature phase-shift keying}
\newacronym{rms}{RMS}{root-mean-square}
\newacronym{sinr}{SINR}{signal-to-interference-and-noise ratio}
\newacronym{snr}{SNR}{signal-to-noise ratio}
\newacronym{tdma}{TDMA}{time-division multiple access}
\newacronym{up}{UP}{user plane}
\newacronym{urllc}{URLLC}{ultra-reliable low-latency communication}
\newacronym{fp}{FP}{fractional programming}